\documentclass[a4paper]{article}

\usepackage{comment}
\usepackage{amsthm}
\usepackage{fancyhdr}
\usepackage{amsmath}
\usepackage{amssymb}
\usepackage[mathscr]{euscript}
\usepackage{graphicx}
\usepackage{subfig}
\usepackage{enumitem}
\usepackage[vertfit]{breakurl}
\usepackage[backend=bibtex,maxbibnames=99,giveninits=true]{biblatex}
\usepackage{hyperref}
\hypersetup{colorlinks=true}
\usepackage{cleveref}

\newcommand{\C}{\mathscr C}

\newcommand{\F}{\mathscr F}

\def\e{{\varepsilon}}

\newcommand{\B}{\mathcal{B}}
\newcommand{\Z}{\mathbb{Z}}

\newcommand{\A}{ \mathscr{A}}

\newcommand{\Pb}{ \mathbb{P}}

\newcommand{\w}{\omega}

\newcommand{\pk}{\mathsf{pk}}
\newcommand{\sk}{\mathsf{sk}}
\newcommand{\rec}{\mathsf{rec}}

\DeclareMathOperator{\lcm}{lcm}
\DeclareMathOperator{\MtA}{MtA}
\DeclareMathOperator{\MtAwc}{MtAwc}
\DeclareMathOperator{\Com}{Com}
\DeclareMathOperator{\Ver}{Ver}

\theoremstyle{plain} 

\newtheorem{teo}{Theorem}[section]
\newtheorem{lem}{Lemma}[section]

\theoremstyle{definition} 

\newtheorem{defi}{Definition}[section]
\newtheorem{oss}{Observation}

\bibliography{main.bib}

\begin{document}

  %
  %
  %
  %
  %
  %
  %
  %
  %
\title{Threshold ECDSA with an Offline Recovery Party}

\author{Michele Battagliola\thanks{michele.battagliola@unitn.it} \and Riccardo Longo\thanks{riccardolongomath@gmail.com} \and Alessio Meneghetti\thanks{alessio.meneghetti@unitn.it} \and Massimiliano Sala\thanks{maxsalacodes@gmail.com}}
   
%
%
%
%
%

%
  

\date{}
\maketitle

\centerline{ Department of Mathematics, University Of Trento}
   \centerline{38123 Povo, Trento, Italy}

\begin{abstract}
    A $(t,n)-$ threshold signature scheme enables distributed signing among $n$ players such that any subset of size at least $t$ can sign, whereas any subset with fewer players cannot.
    Our goal is to produce digital signatures that are compatible with an existing centralized signature scheme: the key-generation and signature algorithms are replaced by a communication protocol between the players, but the verification algorithm remains identical to that of a signature issued using the centralized algorithm. 
    Starting from the threshold scheme for the ECDSA signature due to R. Gennaro and S. Goldfeder, we present the first protocol that supports multiparty signatures with an offline participant during the key-generation phase and that does not rely on a trusted third party.
    \\
    Under standard assumptions on the underlying algebraic and geometric problems (e.g. the Discrete Logarithm Problem for an elliptic curve and the computation of $e^{th}$ root on semi-prime residue rings), we prove our scheme secure against adaptive malicious adversaries.
\end{abstract}
\section{Introduction}
A $(t, n)$-threshold signature scheme enables distributed signing among $n$ players such that any subset of size $t$ can sign, whereas any subset with fewer players cannot.
The first Threshold Multi-Party Signature Scheme  was a protocol for ECDSA signatures proposed by Gennaro et al.~\cite{gennaro1996robust} where $t+1$ parties out of $2t + 1$ were required to sign a message.
Later MacKenzie and Reiter proposed and then improved another scheme~\cite{mackenzie2001two,mackenzie2004two}, which has later been furthermore enhanced~\cite{doerner2018secure,doerner2019threshold,lindell2017fast}.
The first scheme supporting a general $(t, n)$-threshold was proposed in~\cite{gennaro2016threshold}, improved in~\cite{boneh2017using} and in~\cite{gennaro2018fast}.
A parallel approach has been taken by Lindell and Nof in~\cite{lindell2018fast}.
In~\cite{kondi2019refresh} the authors introduce a refresh mechanism, for proactive security against the corruption of different actors in time, that does not require all parties to be online, and in~\cite{canetti2020uc} the authors take a similar approach and propose a protocol that streamlines signature generation and include proactive security mechanisms.
Currently there is a large effort of standardization for threshold signatures, as can be seen in \cite{nist_stand}.
\\
The schemes proposed in the previous papers produce signatures that are compatible with an existing centralized signature scheme.
In this context, the key-generation and signature algorithms are replaced by a communication protocol between the parties, while the verification algorithm remains identical to that of a signature issued using the centralized algorithm.

The need for joint signatures arises frequently in the world of cryptocurrencies, where digital signatures determine ownership rights and control over assets, meaning that protection and custody of private keys is of paramount importance.
A particularly sensitive issue is the resiliency against key loss, since there is no central authority that can restore ownership of a digital token once the private key of the 
wallet is lost. Three possible solutions are:
\begin{itemize}
    \item to rely on a trusted third-party custodian that takes responsibility of key management, but this kind of centralization may form single points of failure and juicy targets for criminal takeovers (there have already been plenty of examples of said events in the past~\cite{chohan2018problems}).
    \item to use multi-sig wallets (available for some cryptocurrencies, like Bitcoin~\cite{nakamoto2019bitcoin}) where the signatures are normal ones, but funds may be moved out of that wallet only with a sufficient number of signatures corresponding to a prescribed set of public keys. Unfortunately, this approach is not supported by every cryptocurrency (e.g. Ethereum~\cite{ethereum}) and such wallets are very easily identifiable. 
    \item to distribute the control of the wallet through advanced multi-signature schemes, in particular with threshold-like policies. This solution may be reached by threshold multi-sig compatible with centralised digital signatures, as in the aforementioned work by Gennaro and Goldfeder and related papers.
\end{itemize}
Regarding this last option, there is a potential problem in real-life applications of these protocols: the recovery party (that allows to recover the wallet funds in case of key loss) is usually not willing to sustain the cost of frequent online collaboration.
For example, a bank may safely guard a \emph{piece of the secret key}, but it is inconvenient and quite costly to make the bank participate in the enrollment of every user.

In this paper, following the latter approach, we propose a protocol in which the recovery party is involved only once (in a preliminary set-up), and afterwards it is not involved until a lost account must be recovered. Our protocol may be seen as an adaption of that in \cite{gennaro2018fast}. We prove its security against (adaptive) adversaries by relying on standard assumptions on the underlying algebraic and geometric problems, such as the strong RSA assumption on semi-prime residue rings and the DDH assumption on elliptic curves.

\paragraph{Organization}
We present some preliminaries in \Cref{preliminaries}, we describe our protocol in \Cref{protocol}.
Another practical problem is to derive many keys from a single secret, for example to 
efficiently manage multiple wallets.
Therefore, we provide also an extension of our protocol that can work with key-derivation in \Cref{key-derivation}.
In \Cref{ecdsasecurity} we state the security property that we claim for our protocol, with a proof similar to that in ~\cite{gennaro2018fast}, which however requires several subtle modifications to tackle the off-line situation.
The adaption of our proof to the key-derivation extension is cumbersome but easy and so we do not provide it here.
Finally in \Cref{conclusions} we draw our conclusions.

\section{Preliminaries}\label{preliminaries}
In this section we present some preliminary definitions and primitives that will be used in the protocol and its proof of security.

In the following when we say that an algorithm is \emph{efficient} we mean that it runs in (expected) polynomial time in the size of the input, possibly using a random source.

\subsection{Assumptions}
Our proof is based on two assumptions: the Decisional Diffie-Hellman~\cite{boneh1998decision} (from now on DDH) and the RSA~\cite{rivest1978method} assumption.

\begin{defi}[DDH Assumption]
Let $\mathbb{G}$ be a cyclic group with generator $g$ and order $n$. Let $a, b, c$ be random elements of $\Z_n$. The Decisional Diffie-Hellman assumption, from now on DDH assumption, states that no efficient algorithm can distinguish between the two distributions $(g, g^a, g^b, g^{ab})$ and $(g, g^a, g^b, g^c)$.
\end{defi}
\begin{defi}[RSA Assumption]
Let $N=pq$ with both $p,q$ safe primes. Let $e$ be an integer such
that $e$ and $\phi(N)$ are coprime.
\begin{itemize}
 \item The \emph{RSA assumption} states that given a random element $s \in \Z_N^*$ no efficient algorithm can find $x$ such that $x^e=s \mod N$.
 \item The \emph{Strong RSA assumption} states that given a random element ${s \in \Z_N^*}$ no efficient algorithm can find $x, e_0 \ne 1$ such that $x^{e_0}=s \mod N$.
\end{itemize}
\end{defi}

\subsection{Zero-Knowledge Proofs}
In the protocol various Zero-Knowledge Proofs (ZKP)~\cite{goldreich1986proofs} are used to enforce the respect of the passages prescribed by the specifications.
In fact in the proof of security we can exploit the soundness of these sub-protocols to extract valuable information from the adversary, and their zero-knowledge property to simulate correct executions even without knowing some secrets.
We can do so because we see the adversary as a (black-box) algorithm that we can call on arbitrary input, and crucially we have the faculty to rewind its execution.

In particular we use ZKP \emph{of Knowledge} (ZKPoK) to guarantee the usage of secret values that properly correspond to the public counterpart: the Schnorr protocol for discrete logarithms (see~\cite{schnorr1989efficient} and~\Cref{Schnorr}) and an Integer-Factorization Proof (see~\cite{integer} and~\Cref{fact}) for Paillier keys.
The soundness property of a ZKPoK guarantees that the adversary must know the secret input, and opportune rewinds and manipulations of the adversary's execution during the proof allows us to extract those secrets and use them in the simulation.
Conversely exploiting the zero-knowledge property we can trick the adversary in believing that we know our secrets even if we don't, thus we still obtain a correct simulation of our protocol form the adversary's point of view.

The other ZKP used is a range proof (see~\cite{gennaro2018fast,mackenzie2001two} and~\cref{range}) that guarantees a proper execution of the share conversion protocol of~\Cref{mta}, however the same discussion of~\cite{gennaro2018fast}, Section 5 (about how the security is not significatively affected removing this proof) applies to our protocol.

\subsection{Paillier Cryptosystem}\label{paillier}
In our protocol we use the Paillier cryptosystem, a partially homomorphic asymmetric encryption scheme presented by P. Paillier in \cite{paillier}.
Suppose that Bob wants to send an encrypted message to Alice. The workflow of the algorithm is the following:
\begin{itemize}
\item \textbf{Key-Generation:}
	\begin{enumerate}
		\item Alice chooses two large primes $p,q$ uniformly at random, such that $\gcd(p q,(p-1)(q-1))= 1$.
		\item Alice sets $N=p q$ and $\lambda = \lcm(p-1,q-1)$. 
		\item Alice picks $\Gamma \in \Z^*_{N^2}$ uniformly at random. In this context $\Z^*_{N^2}$ indicates the ring of units of $\Z_{N^2}$.
		\item Alice checks that $N$ divides the order of $\Gamma$. To do so it is sufficient (see~\cite{paillier}, Section 3) to compute $\mu = (L(g^\lambda \mod N^2))^{-1} \mod N$, where $L(x)$ is the quotient of the Euclidean division $\frac{x-1}{N}$, i.e. the largest integer value $k$ such that $x-1 \ge k N$.
		\item The public encryption key is $(\Gamma,N)$. The private encryption key is $(\mu,\lambda)$.
	\end{enumerate}
	\item \textbf{Message Encryption:}
	\begin{enumerate}
		\item To send a message $m \in \Z_N$ to Alice Bob picks $r \in \Z_N^*$.
		\item Bob computes $c=\Gamma^m r^N \mod N^2$.
	\end{enumerate}
	\item \textbf{Decryption}
	\begin{enumerate}
		\item Alice computes $m = L(c^\lambda \mod N^2)\cdot \mu \mod N$.
	\end{enumerate}
	\end{itemize}
Let $E$ and $D$ be the encryption and decryption functions respectively. Given two plaintexts $m_1, m_2 \in Z_N$, their associated ciphertexts $c_1, c_2 \in Z_{N^2}^*$ and $a \in Z_N$ then we define $+_E$ and $\times_E$ as follows:
\begin{align*}
 c_1 +_E c_2 \;&=\; c_1c_2 \mod N^2,\\
  a \times_E c_1 &=\; c_1^a \mod N^2.
\end{align*}
Then the homomorphic properties of the Pallier cryptosystem are:
 \begin{align*}
	 D(c_1+_E c_2) \;&=\; m_1 + m_2 \mod N,\\
	 D(a \times_E c_1) &=\; a m_1 \mod N.
	\end{align*}

\subsection{ECDSA}
The Elliptic Curve Digital Signature Algorithm (ECDSA), presented in~\cite{ecdsa}, is a variant of the Digital Signature Algorithm (DSA)~\cite{kravitz1993digital} which uses elliptic curve cryptography.

Suppose Alice wants to send a signed message $m$ to Bob.
Initially, they agree on a cryptographic hash function~\cite{rogaway2004cryptographic} $H$, an elliptic curve $E$, a base point $\B$ for $E$, with $n$ the order of $\B$ a prime.
For any point $\mathcal{P} \in E$ we use the notation $\mathcal{P}_x$ to denote the value of the first coordinate of the point $\mathcal{P}$.
\\\pagebreak[3]
\\The protocol works as follows:
\begin{enumerate}
 \item Alice creates a key-pair consisting of a private integer $d$, selected uniformly at random in the interval $[1,n-1]$ and the public point $\mathcal{Q} = d\B$.
 \item Alice computes $e=H(m)$ and picks $k$ uniformly at random in the interval $[1,n-1]$.
 \item Alice computes the point $\mathcal{R} = k^{-1}\B$.
 \item Alice computes $s=k(e+rd)$ with $r=\mathcal{R}_x$.
 \end{enumerate}
The signature is the pair $(r,s)$.
\\
\\
To verify the signature Bob performs the following steps:
\begin{enumerate}
 \item Bob checks that $r,s \in [1,n-1]$,
 \item Bob computes $e= H(m)$,
 \item Bob computes $u_1 = es^{-1} \mod n$ and $u_2 = rs^{-1} \mod n$,
 \item Bob computes the point $\mathcal{U} = u_1\B + u_2\mathcal{Q} $,
 \item checks that $r \equiv \mathcal{U}_x \mod n$.
\end{enumerate}

\subsection{A Share Conversion Protocol}\label{mta}
Assume that we have two parties, Alice and Bob, holding two secrets, respectively $a, b \in \Z_q$ with $q$ prime, such that $ab = x \mod q$. We can imagine $a, b$ as private shards of a shared secret $x$. Alice and Bob would like to perform a multiplicative-to-additive conversion in order to compute $\alpha, \beta \in \Z_q$ such that $\alpha + \beta = x \mod q$.
In this section we will present a protocol for this based on the Paillier Encryption Scheme (see~\Cref{paillier}). We assume that Alice has a public key $A=(N,\Gamma)$, and $E_A$ will indicate the Pallier encryption with A. Moreover we need a value $K > q$.
\begin{enumerate}
 \item Alice initiates the protocol:
 \begin{itemize}
  \item Alice computes $\mathtt{c}_A = E_A(a)$ and sends it to Bob.
  \item Alice proves in ZK that $a <K$ via the first range proof explained in \Cref{range}.
 \end{itemize}
 \pagebreak[3]
 \item Bob generates his shard $\beta$:
 \begin{itemize}
  \item Bob computes $\mathtt{c}_B = b \times_E \mathtt{c}_A +_E E_A(\beta') = E_A(ab + \beta')$ where $\beta'$ is chosen uniformly at random in $\Z_N$.
  \item Bob sends $\mathtt{c}_B$ to Alice.
  \item Bob proves in ZK that $b<K$ via the second range proof presented in \Cref{range}.
  \item If $B=g^b$ is public Bob proves in ZK that he knows $b, \beta'$ such that $g^b=B$ and $\mathtt{c}_B = b \times_E \mathtt{c}_A +_E E_A(\beta').$
 \end{itemize}
 \item Alice decrypts $\mathtt{c}_B$ to obtain $\alpha'$.
 \item Alice obtains her shard $\alpha = \alpha' \mod q$, Bob obtains his shard $\beta = -\beta'$.
 \end{enumerate}
The protocol takes two different names depending on whether $B = g^b$ is public or not. In the first case we refer to this protocol as MtAwc (Multiplicative to Additive with check), because Bob performs the extra check at the end, in the second we refer to it simply as MtA.

For more details about the protocol and the security proof see~\cite{gennaro2018fast}.

\subsection{Feldman-VSS}\label{feldman}
Feldman's VSS scheme\cite{paperFeldman} is a verifiable secret sharing scheme built on top of Shamir's scheme\cite{shamirSS}. A secret sharing scheme is verifiable if auxiliary information is included, that allows players to verify the consistency of their shares.
We use a simplified version of Feldman's protocol: if the verification fails the protocol does not attempt to recover excluding malicious participants, instead it aborts altogether.
In a sense we consider \emph{somewhat honest} participants, for this reason we do not need stronger schemes such as~\cite{gennaro1999secure,schoenmakers1999simple}.
\\The scheme works as follows:
\begin{enumerate}
 \item A cyclic group $\mathbb{G}$ of prime order $p$ is chosen, as well as a generator $g \in \mathbb{G}$. The group $\mathbb{G}$ must be chosen such that the discrete logarithm is hard to compute.
 \item The dealer computes a random polynomial $P$ of degree $t$ with coefficients in $\Z_p$, such that $P(0)=s$ where $s$ is the secret to be shared.
 \item Each of the $n$ share holders receive a value $P(1),...,P(n) \mod p$. So far, this is exactly Shamir's scheme. 
 \item To make these shares verifiable, the dealer distributes commitments to the coefficients of $p$. Let $P(X)=s+\sum_{i=1}^n a_i X^i$, then the commitments are $c_0=g^s$ and $c_i=g^{a_i}$ for $i>0$.
 \item Any party can verify its share in the following way: let $\alpha$ be the share received by the $i$-th party, then it can check if $\alpha=P(i)$ by checking if the following equality holds:
 $$g^\alpha = \prod_{j=0}^t c_j^{(i^j)} = g^s \prod_{j=1}^t g^{a_j(i^j)} = g^{s + \sum_{j=1}^t a_j(i^j)}=g^{P(i)}.$$
\end{enumerate}
In the proof we will need to simulate a $(2,2)$-threshold instance of this protocol without knowing the secret value $s$.

Let us use an additive group with generator $\B$, and let $\mathcal{Y} = s \B$, the simulation proceeds as follows:
\begin{itemize}
 \item the dealer selects two random values $a,b$ and forces $P(1) = a$, $P(2) = b$;
 \item then it computes:
	\begin{align}\label{ECVSS}
	 c_1&=(a\B-\mathcal{Y}),\\
	 c_2&=\frac{1}{2}(b\B-\mathcal{Y});
	\end{align}
	\item the other players can successfully verify their shards, checking that
	\begin{align}
	 a\B &= \mathcal{Y} + c_1 = \mathcal{Y} +a\B-\mathcal{Y},\\
	 b\B &= \mathcal{Y} + 2 c_2 = \mathcal{Y}+2\cdot\frac{1}{2}(b\B-\mathcal{Y}).
	\end{align}
\end{itemize}

\section{Protocol Description}\label{protocol}
In this section we describe the details of our protocol.
After some common parameters are established, one player chooses a long-term asymmetric key and then can go offline, leaving the proper generation of the signing key to the remaining two participants. 
For this reason the signature algorithm is presented in two variants, one used jointly by the two players who performed the Key-Generation, and one used by the offline player and one of the others.
\\More specifically the protocol is comprised by four phases:
\begin{enumerate}
 \item \textbf{Setup Phase} (\Cref{setup}): played by all the parties, it is used to decide common parameters. Note that in many contexts these parameters are mandated by the application, so the parties merely acknowledge them, possibly checking they respect the required security level.
 \item \textbf{Key-Generation} (\Cref{key-gen}): played by only two parties, from now on $P_1$ and $P_2$. It is used to create a public key and the private shards for each player.
 \item \textbf{Ordinary Signature} (\Cref{ordinarysignature}): played by $P_1$ and $P_2$. As the name suggests this is the normal use-case of the protocol.
 \item \textbf{Recovery Signature} (\Cref{recoverysignature}): played by $P_3$ and one between $P_1$ and $P_2$. This models the unavailability of one player, with $P_3$ stepping up as a replacement.
\end{enumerate}
From here on with the notation ``$P_i$ does something'', we mean that both $P_1$ and $P_2$ perform the prescribed task independently.
Similarly, the notation ``$P_i$ sends something to $P_j$'' means that $P_1$ sends to $P_2$ and $P_2$ sends to $P_1$.

\subsection{Setup Phase}\label{setup}
 This phase involves all the participants and is used to decide the parameters of the algorithm.
 \\The parameters involved are the following:
 \\
 \\
 \begin{tabular}{|ll|}
 	\hline
 	\textbf{Player 1 and 2} & \\ \hline
 	\textbf{Input:} & $-$ \\
 	\textbf{Private Output:} & $-$\\
 	\textbf{Public Output:} & $E, \B, q, H$\\
 	\hline
 \end{tabular}
 \begin{tabular}{|ll|}
 	\hline
 	\textbf{Player 3} & \\ \hline
 	\textbf{Input:} & $-$ \\
 	\textbf{Private Output:} & $\sk_3$\\
 	\textbf{Public Output:} & $\pk_3$\\
 	\hline
 \end{tabular}
 \\
 \\
 $P_3$ chooses an asymmetric encryption algorithm and a key pair $(\pk_3, \sk_3)$, then it publishes $\pk_3$, keeping $\sk_3$ secret. $\pk_3$ is the key that $P_1$ and $P_2$ will use to communicate with $P_3$. 
 The algorithm which generates the key pair ($\sk_3$, $\pk_3$) and the encryption algorithm itself are unrelated to the signature algorithm, but it is important that both of them are secure.
 \\More formally we require that the encryption protocol has the property of IND-CPA \cite{bellare2005introduction,indcpa}:
 \begin{defi}[IND-CPA]
  Let $\Pi = (\mathsf{Gen}, \mathsf{Enc}, \mathsf{Dec})$ be a public key encryption scheme. Let us define the
following experiment between an adversary $\A$ and a challenger $\mathscr{C}^b$ parameterized by a bit $b$:
\begin{enumerate}
\item The challenger runs $\mathsf{Gen}(1^k)$ to get $\sk$ and $\pk$, the secret and public keys. Then it gives $\pk$ to $\A$.
\item $\A$ outputs two messages $(m_0, m_1)$ of the same length.
\item The challenger computes $\mathsf{Enc}(\pk, m_b)$ and gives it to $\A$.
\item $\A$ outputs a bit $b'$(if it aborts without giving any output, we just set $b'=0$). The challenger returns $b'$ as the output of the game.
\end{enumerate}
We say that $\Pi$ has the property of being \emph{indistinguishable under chosen plaintext attacks} (IND-CPA) or simply  \emph{secure against a chosen plaintext attack}, if for any $k$ and any probabilistic polynomial time adversary $\A$ the function
\begin{equation}
 \mathsf{Adv}(\A) = \Pb[\mathscr{C}^1(\A, k) = 1] - \Pb[\mathscr{C}^0(\A, k)= 1],
\end{equation}
i.e. $\mathsf{Adv}(\A) = \Pb[b'=b] - \Pb[b'\ne b]$, is negligible.
 \end{defi}
 Then $P_1$ and $P_3$ need to agree on a secure hash function $H$, an elliptic curve $E$ with group of points of prime order $q$, and a generator $\B \in E$ of said group. The order identifies the ring $\Z_q$ used for scalar values.
 
\subsection{Key-Generation}\label{key-gen}
\begin{tabular}{|ll|}
 	\hline
 	\textbf{Player 1} & \\ \hline
 	\textbf{Input:} & $\pk_3$ \\
 	\textbf{Private Output:} & $p_1, q_1, \w_1$\\
 	\textbf{Shared Secrets:} 
 	& $\rec_{1,3}, \rec_{2,3}, d$\\
 	\textbf{Public Output:} & $\Gamma_1, N_2, \Gamma_2, \mathcal{Y}$\\
 	\hline
 \end{tabular}
 \begin{tabular}{|ll|}
	\hline
	\textbf{Player 2} & \\ \hline
	\textbf{Input:} & $\pk_3$ \\
	\textbf{Private Output:} & $p_2, q_2, \w_2$\\
 	\textbf{Shared Secrets:}
	&$\rec_{1,3}, \rec_{2,3}, d$\\
	\textbf{Public Output:} & $\Gamma_2, N_1, \Gamma_1, \mathcal{Y}$\\
	\hline
\end{tabular}

The protocol proceeds as follows:
\begin{enumerate}
	\item Secret key-generation and communication:
	\begin{enumerate}[label=\alph*.]
 		\item $P_i$ generates a Paillier public key $(N_i, \Gamma_i)$ and its corresponding secret key $(p_i, q_i)$.
 		\item $P_i$ selects randomly $u_i, \sigma_{3, i} \in \Z_q$.
 		\item $P_i$ computes $[\mathsf{KGC}_i,\mathsf{KGD}_i]:=\Com(u_i\B)$.
 		\item $P_i$ computes $[\mathsf{KGCS}_i,\mathsf{KGDS}_i]:=\Com(\sigma_{3, i}\B)$.
 		\item $P_i$ sends $\mathsf{KGC}_i$, $\mathsf{KGCS}_i$, and $(N_i, \Gamma_i)$ to $P_j$.
 		\item $P_i$ sends $\mathsf{KGD}_i,\mathsf{KGDS}_i$ to $P_j$.
 		\item $P_i$ computes $u_j\B=\Ver(\mathsf{KGC}_j, \mathsf{KGD}_j)$ and $\sigma_{3, j}\B = \Ver(\mathsf{KGCS}_j, \mathsf{KGDS}_j)$.
 	\end{enumerate}
 \item Feldman VSS protocol and generation of $P_3$'s data:
	\begin{enumerate}[label=\alph*.]
		\item $P_i$ selects randomly $m_i \in \Z_q$. 
		\item $P_i$ sets $f_i(X) = u_i+m_i X$, and $\sigma_{i,1}=f_i(2), \sigma_{i,2}=f_i(3)$, ${\sigma_{i,3}=f_i(1)}$. Then $P_i$ computes and distributes the shards $c_{i,j}$ for the Feldman-VSS, as described in \Cref{feldman}.
		\item Everyone checks the integrity and consistency of the shards according to the VSS protocol.
		\item $P_i$ encrypts $\sigma_{i,3}$ and $\sigma_{3,i}$ with $\pk_3$ to obtain $\rec_{i,3}$.
		\item $P_i$ sends $\sigma_{i,j}, m_i\B, \rec_{i,3}$ to $P_j$.
	\end{enumerate}
 	\item $P_i$ computes the private key $x_i = \sigma_{1,i} + \sigma_{2,i} + \sigma_{3,i}$.
 	\item $P_i$ proves in ZK that it knows $x_i$ using Schnorr's protocol of \Cref{Schnorr}.
 	\item $P_i$ proves in ZK that it knows $p_i, q_i$ such that $N_i = p_iq_i$ using the integer factorization ZKP of \Cref{fact}.
 	\item Public key-generation and shares conversion:
 	\begin{enumerate}[label=\alph*.]
 		\item the public key is $\mathcal{Y}=\sum_{i=1}^3 u_i\B$, where $u_3\B=3(\sigma_{3,1}\B) - 2 (\sigma_{3,2}\B)$. So $u_3 = 3\sigma_{3,1} - 2 \sigma_{3,2}$. From now on we will set $ u = \sum_{i=1}^3 u_i$. Obviously $u\B = \mathcal{Y}$.
 		\item the private key of $P_1$ is $\w_1= 3x_1$, while the private key of $P_2$ is $w_2=-2x_2$. We can observe that $w_1 + w_2 = u$.
 		\item $P_1$ and $P_2$ can compute the common secret $d=(\sigma_{2,1}\sigma_{2,3}\B)_x$ that will be used for key derivation.
 	\end{enumerate}
\end{enumerate}
\begin{oss}\label{ossu3}
 We define $u_3 = 3\sigma_{3,1} - 2 \sigma_{3,2}$ because we need to be consistent with the Feldman-VSS protocol. Indeed, suppose that $\sigma_{3,2}$ and $\sigma_{3,1}$ are valid shards of a Feldman-VSS protocol where the secret is $u_3$. Since there is an $m_3$ such that $\sigma_{3,2} = u_3+ m_3 \cdot 3$ and $\sigma_{3,1} = u_3+ m_3 \cdot 2$, we have that: $$3\sigma_{3,1}-2\sigma_{3,2} = 3u_3+6m_3-2u_3-6m_3=u_3.$$
 Note that $u_3\B$ can be computed by both $P_1$ and $P_2$, but they cannot compute $u_3$.
 \end{oss}
 
\subsection{Signature Algorithm}\label{signature}
This protocol is used by two players, called $P_A$ and $P_B$, to sign messages.
$P_1, P_2$, and $P_3$ take the role of either $P_A$ or $P_B$ depending on the situation, see \Cref{ordinarysignature,recoverysignature}.

The participants agree on a message $M$ to sign.
The goal of this protocol is to produce a valid ECDSA signature $(r,s)$ for the public key $y$.\\
The parameters involved are:
\\
\\
\begin{tabular}{|ll|}
 \hline
 \textbf{Player A} & \\ \hline
 \textbf{Input:} & $M, \w_A, \Gamma_A, p_A,$ \\
 & $q_A, N_B, \Gamma_B, \mathcal{Y}$ \\
 \textbf{Public Output:} & $(r,s)$\\
 \hline
\end{tabular}
\begin{tabular}{|ll|}
 \hline
 \textbf{Player $B$} & \\ \hline
 \textbf{Input:} & $M, \w_B, \Gamma_B, p_B,$\\
 &$ q_B, N_A, \Gamma_A, \mathcal{Y}$ \\
 \textbf{Public Output:} & $(r,s)$\\
 \hline
\end{tabular}
\pagebreak[3]
\\The protocol works as follows:
 \begin{enumerate}
 	\item Commitment phase:
 	\begin{enumerate}[label=\alph*.]
 	\item $P_i$ picks randomly $k_i, \gamma_i \in \Z_q$. 
 \item $P_i$ computes $\mathcal{G}_i = \gamma_i \B$ and $[\Delta_i, D_i] = \Com(\mathcal{G}_i)$. \\We define $\gamma = \gamma_A+\gamma_B$ and $k=k_A+k_B$.
 	\item $P_i$ sends $\Delta_i$ to $P_j$.
 	\end{enumerate}
 	\item Multiplicative to additive conversion:
 	\begin{enumerate}[label=\alph*.]
 		\item $P_i$ re-computes the parameters for Paillier encryption from its keys: $\lambda_i = \lcm(p_i-1,q_i-1), {\mu_i=(L_i(\Gamma_i^{\lambda_i} \mod N_i^2))^{-1} \mod N_i}$ where $L_i(x) = \frac{x-1}{N_i}$.
 		\item $P_A$ and $P_B$ run $\MtA(k_A, \gamma_B)$ to get respectively $\alpha_{A,B}, \beta_{A,B}$ such that $k_A\gamma_B=\alpha_{A,B} + \beta_{A,B}$. They also run it on $k_B, \gamma_A$ to get respectively $\beta_{B,A}$ and $\alpha_{B,A}$.
 		\item $P_i$ sets $\delta_i=k_i\gamma_i+\alpha_{i,j} + \beta_{j,i}$.
 		\item $P_A$ and $P_B$ run $\MtAwc(k_A, \w_B)$ to get respectively $\mu_{A,B}$ $\nu_{A,B}$ such that $k_A\w_B = \mu_{A,B} +\nu_{A,B}$. They also run it on $k_B, \w_A$ to get respectively $\nu_{B,A}$ and $\mu_{B,A}$.
 		\item $P_i$ sets $\sigma_i= k_i \w_i + \mu_{ij} + \nu_{ji}$. We can observe that $\sum \sigma_i = ku$.
 		\item $P_i$ sends $\delta_i$ to $P_j$.
 		\item $P_A$ and $P_B$ compute $\delta = \delta_A+\delta_B$ and $\delta^{-1} \mod q$.
 	\end{enumerate}
 \item Decommitment phase and ZKP:
 \begin{enumerate}[label=\alph*.]
 	\item $P_i$ sends $D_i$.
 	\item $P_i$ computes $\mathcal{G}_j = \Ver(\Delta_j,D_j)$.
 	\item $P_i$ proves in ZK that it knows $\gamma_i$ such that $\gamma_i\B=\mathcal{G}_i$, using Schnorr's protocol.
 	\item $P_A$ and $P_A$ set $\mathcal{R} = \delta^{-1}(\mathcal{G}_A+\mathcal{G}_B)$ and $r=\mathcal{R}_x$. We can observe that $\delta = \gamma k$ and $\mathcal{G}_A+\mathcal{G}_B = \gamma\B$, so $\mathcal{R}=k^{-1}\B$.
 \end{enumerate}
\item Signature generation:\label{ECsgigen}
\begin{enumerate}[label=\alph*.]
 	\item Both players set $m=H(M)$.
 	\item $P_i$ computes $s_i=mk_i+r\sigma_i$.
 	\item\label{ECsigcheck_com} $P_i$ picks uniformly at random $l_i, \rho_i \in \Z_q$, and computes $\mathcal{W}_i:=s_i\mathcal{R}+l_i\B$, $\mathcal{Z}_i=\rho_i\B$, and then $[\hat{\Delta}_i, \hat{D}_i] = \Com(\mathcal{W}_i,\mathcal{Z}_i)$.
 	\item $P_i$ sends $\hat{\Delta}_i$ to $P_j$. 
 	\item $P_i$ sends $\hat{D}_i$ to $P_j$.
 	\item $P_i$ computes $[\mathcal{W}_j,\mathcal{Z}_j] := \Ver(\hat{\Delta}_j,\hat{D}_j)$.
 	\item Each $P_i$ proves in ZK that it knows $s_i,l_i,\rho_i$ such that $\mathcal{W}_i = s_i R + l_i \B$ and $\mathcal{Z}_i = \rho_i \B$ (if a ZKP fails, the protocol aborts).
 	\item $P_A$ and $P_A$ compute $\mathcal{W} := -m \B -r y + \mathcal{W}_A + \mathcal{W}_B$ and $\mathcal{Z} := \mathcal{Z}_A + \mathcal{Z}_B$.
 	\item Each $P_i$ computes $\mathcal{U}_i := \rho_i \mathcal{W}$, $\mathcal{T}_i := l_i \mathcal{Z}$ and $[\tilde{\Delta}_i,\tilde{D}_i] := \Com(\mathcal{U}_i, \mathcal{T}_i)$.
 	\item $P_i$ sends $\tilde{\Delta}_i$ to $P_j$.
 	\item $P_i$ sends $\tilde{D}_i$ to $P_j$.
 	\item $P_i$ computes $[\mathcal{U}_j,\mathcal{T}_j] := \Ver(\tilde{\Delta}_j,\tilde{D}_j)$.
 	\item If $\mathcal{T}_1 + \mathcal{T}_2 \ne \mathcal{U}_1 + \mathcal{U}_2$ the protocol aborts.
 	\item $P_i$ sends $s_i$.
 	\item $P_1$ and $P_2$ compute $s := s_1 + s_2$.
 	\item If $(r,s)$ is not a valid signature, the players abort, otherwise they accept and end the protocol.
\end{enumerate}
\end{enumerate}
\subsection{Ordinary Signature}\label{ordinarysignature}
This is the case where $P_1$ and $P_2$ wants to sign a message $m$. They run the signature algorithm of \Cref{signature} with the following parameters (supposing $P_1$ play the roles of $P_A$ and $P_2$ of $P_B$):
\\
\\
 \begin{tabular}{|ll|}
 	\hline
 	\textbf{Player A} & \\ \hline
 	\textbf{Input:} & $M, \w_1, \Gamma_1, p_1,$ \\
 	& $q_1, N_2, \Gamma_2, \mathcal{Y}$ \\
 	\textbf{Public Output:} & $(r,s)$\\
 	\hline
 \end{tabular}
 \begin{tabular}{|ll|}
	\hline
	\textbf{Player $B$} & \\ \hline
	\textbf{Input:} & $M, \w_2, \Gamma_2, p_2,$\\
	&$ q_2, N_1, \Gamma_1, \mathcal{Y}$ \\
	\textbf{Public Output:} & $(r,s)$\\
	\hline
 \end{tabular}
\subsection{Recovery Signature}\label{recoverysignature}
If one between $P_1$ and $P_2$ is unable to sign, then $P_3$ has to come back online and a recovery signature is performed.

We have to consider two different cases, depending on who is offline.
First we consider the case in which $P_2$ is offline, therefore $P_1$ and $P_3$ sign.\\
The parameters involved are:
\\
\\
\begin{tabular}{|ll|}
	\hline
	\textbf{Player 1} & \\ \hline
	\textbf{Input:} & $M, \w_1, \Gamma_1, p_1, q_1, \mathcal{Y}$ \\
	& $\rec_{2,3}, \rec_{1,3}$\\
	\textbf{Public Output:} & $(r,s)$\\
	\hline
\end{tabular}
\begin{tabular}{|ll|}
	\hline
	\textbf{Player 3} & \\ \hline
	\textbf{Input:} & $M,\sk_3$ \\
	 &\\
	\textbf{Public Output:} & $(r,s)$\\
	\hline
\end{tabular}
\\The workflow in this case is:
\begin{enumerate}
	\item Communication:
	\begin{enumerate}[label=\alph*.)]
		\item $P_1$ contacts $P_3$, which comes back online.
		\item $P_1$ sends $y$ and $\rec_{1,3}, \rec_{2,3}$ to $P_3$.
	\end{enumerate}
	\item Paillier keys generation and exchange:
	\begin{enumerate}[label=\alph*.)]
		\item $P_3$ generates a Paillier public key $(N_3, \Gamma_3)$ and its relative secret key $(p_3, q_3)$.
		\item $P_i$ sends $N_i, \Gamma_i$ to the other party.
		\item $P_i$ proves to $P_j$ that it knows $p_i, q_i$ such that $N_i = p_i q_i$ using integer factorization ZKP.
	\end{enumerate}
	\item $P_3$'s secrets generation:
	\begin{enumerate}[label=\alph*.)]
		\item $P_3$ decrypts $\rec_{1,3}$ and $\rec_{2,3}$ with its private key $\sk_3$, getting $\sigma_{1,3}, \sigma_{3,1}$, $\sigma_{2,3}, \sigma_{3,2}$.
		\item $P_3$ computes $x_3=\sigma_{1,3}+2\sigma_{3,1}-\sigma_{3,2} +\sigma_{2,3}$.
		\item $P_i$ proves in ZK that it knows $x_i$ using Schnorr's protocol.
	\end{enumerate}
	\item Signature generation:
	\begin{enumerate}[label=\alph*.)]
		\item $P_1$ computes $\tilde{\w}_1:=-\frac{1}{3}\w_1$.
		\item $P_3$ computes $\w_3:=2x_3$.
		\item $P_1$ and $P_3$ perform the Signature Algorithm of \Cref{signature} as $P_A$ and $P_B$ respectively, where the $P_1$ uses $\tilde{\w}_1$ in place of $\w_A$ and $P_3$ uses $\w_3$ in place of $\w_B$ (the other parameters are straightforward).
	\end{enumerate}
\end{enumerate}
\vspace*{.5cm}
We consider now the second case in which $P_1$ is offline, therefore $P_2$ and $P_3$ sign.
\\The parameters involved are:
\\
\\
\begin{tabular}{|ll|}
	\hline
	\textbf{Player 2} & \\ \hline
	\textbf{Input:} & $M, \w_2, \Gamma_2, p_2, q_2, \mathcal{Y}$ \\
	& $\rec_{2,3}, \rec_{1,3}$\\
	\textbf{Public Output:} & $(r,s)$\\
	\hline
\end{tabular}
\begin{tabular}{|ll|}
	\hline
	\textbf{Player 3} & \\ \hline
	\textbf{Input:} & $M, \sk_3$ \\
	&\\
	\textbf{Public Output:} & $(r,s)$\\
	\hline
\end{tabular}
\\
\\
The workflow of this case is:
\begin{enumerate}
	\item Communication:
	\begin{enumerate}[label=\alph*.)]
		\item $P_2$ contacts $P_3$, which comes back online.
		\item $P_2$ sends $y$ and $\rec_{1,3}, \rec_{2,3}$ to $P_3$.
	\end{enumerate}
	\pagebreak[3]
	\item Paillier keys generation and exchange:
	\begin{enumerate}[label=\alph*.)]
		\item $P_3$ generates a Paillier public key $(N_3, \Gamma_3)$ and its relative secret key $(p_3, q_3)$.
		\item $P_i$ sends $N_i, \Gamma_i$ to the other party.
		\item $P_i$ proves to $P_j$ that it knows $p_i, q_i$ such that $N_i = p_i q_i$ using integer factorization ZKP.
	\end{enumerate}
	\item $P_3$'s secrets generation:
	\begin{enumerate}[label=\alph*.)]
		\item $P_3$ decrypts $\rec_{1,3}$ and $\rec_{2,3}$ with its private key $\sk_3$, getting $\sigma_{1,3}$, $\sigma_{3,1}, \sigma_{2,3}, \sigma_{3,2}$.
		\item $P_3$ computes $x_3=\sigma_{1,3}+2\sigma_{3,1}-\sigma_{3,2} +\sigma_{2,3}$.
			\item $P_i$ proves in ZK that it knows $x_i$ using Schnorr's protocol.
	\end{enumerate}
	\item Signature generation:
	\begin{enumerate}[label=\alph*.)]
		\item $P_2$ computes $\tilde{\w}_2:=\frac{1}{4}\w_2$.
		\item $P_3$ computes $\w_3:=\frac{3}{2}x_3$.
		\item $P_2$ and $P_3$ perform the Signature Algorithm of \Cref{signature} as $P_A$ and $P_B$ respectively, where the $P_2$ uses $\tilde{\w}_2$ in place of $\w_A$ and $P_3$ uses $\w_3$ in place of $\w_B$ (the other parameters are straightforward).
	\end{enumerate}
\end{enumerate}
\begin{oss}
 We define $\tilde{\w}_i$ in this way since we need $\tilde{\w}_i+ \w_3 = u$, with $i \in \{1,2\}$. Moreover we define $x_3$ in this way for the same reasons explained in Observation \ref{ossu3}.
\end{oss}
 \subsection{Key Derivation}\label{key-derivation}
 In many applications it is useful to deterministically derive multiple signing keys from a single \emph{master} key (see e.g. BIP32 for Bitcoin wallets~\cite{wuille_2012}), and this practice becomes even more useful since in our protocol the key-generation is a multi-party computation.
 
 In order to perform the key derivation we need a derivation index $i$ and the common secret $d$ created during the Key-Generation protocol.
\\The derivation is performed as follows:
\begin{itemize}
\item $P_1$ and $P_2$ perform the key derivation:
\begin{itemize}
\item $\omega_1 \to \omega_1^i = \omega_1 + 3H(d||i),$
\item $\omega_2 \to \omega_2^i = \omega_2 - 2H(d||i),$
\item $\mathcal{Y} \to \mathcal{Y}^i = y+H(d||i)\B.$
\end{itemize}
\item $P_1$ and $P_3$ perform the key derivation:
\begin{itemize}
\item $\omega_1 \to \omega_1^i = \omega_1 - H(d||i),$
\item $\omega_3 \to \omega_3^i = \omega_3 + 2H(d||i),$
\item $\mathcal{Y} \to \mathcal{Y}^i = y+H(d||i)\B.$
\end{itemize}
\item $P_2$ and $P_3$ perform the key derivation:
\begin{itemize}
\item $\omega_2 \to \omega_2^i = \omega_2 - \frac{1}{2} H(d||i),$
\item $\omega_3 \to \omega_3^i = \omega_3 + \frac{3}{2}H(d||i),$
\item $\mathcal{Y} \to \mathcal{Y}^i = y+H(d||i)\B.$
\end{itemize}
\end{itemize}
\begin{oss}
 We observe that the algorithm outputs valid keys, such that, for example:
 $$ (\w_1^i+\w_2^i)\B = y^i.$$
 Since $(\w_1^i+\w_2^i) = \w_1 + \w_2 + H(d||i)$ we have that:
 $$ (\w_1^i+\w_2^i)\B = (\w_1 + \w_2 + H(d||i))\B = \mathcal{Y} + H(d||i)\B = \mathcal{Y}^i.$$
 With the same procedure we can prove that also the other pairs of derived keys are consistent.
\end{oss}

\section{Security Proof} \label{ecdsasecurity}
As customary for digital signature protocols, we state the security of our scheme as an \emph{unforgeability} property, defined as follows (adapted from the classical definition introduced in~\cite{goldwasser1988digital}):
\begin{defi}
	We say that a $(t, n)$-threshold signature scheme is unforgeable if no malicious adversary who corrupts at most $t-1$ players can produce with non-negligible probability the signature on a new message $m$, given the view of \textbf{Threshold-Sign} on input messages $m_1,...,m_k$ (which the adversary adaptively chooses), as well as the signatures on those messages.
\end{defi}
Referring to this definition the security of our protocol derives from the following theorem, whose proof is the topic of this section:
 \begin{teo}\label{theorem_security}
 	Assuming that
 	\begin{itemize}
 	\item the ECDSA signature scheme is unforgeable,
 	\item the strong RSA assumption holds,
 	\item $(\Com, \Ver)$ is a non malleable commitment scheme,
 	\item the DDH assumption holds,
 	\item and that encryption algorithm used by $P_3$ is IND-CPA \footnote{In this proof we focus on the unforgeability property. We discuss other security aspects, such as recovery resiliency, in~\Cref{conclusions}.},
 	\end{itemize}
 the threshold ECDSA protocol is unforgeable.
 \end{teo}

The proof will use a classical game-based argument, our goal is to show that if there is an adversary $\mathscr{A}$ that forges the threshold scheme with a non-negligible probability $\e>\lambda^{-c}$, for a polynomial $\lambda(x)$ and $c>0$, we can build a forger $\F$ that forges the centralised ECDSA scheme with non-negligible probability as well.

Since the algorithm presented is a $(2,3)$-threshold signature scheme the adversary will control one player and $\F$ will simulate the remaining two. Since the role of $P_3$ is different we have to consider two distinct cases: one for $\A$ controlling $P_3$ and one for $\A$ controlling one between $P_1$ and $P_2$ (whose roles are symmetrical). The second case is way more interesting and difficult, so it will be discussed first, and for now we suppose without loss of generality that $\A$ controls $P_2$.

\begin{defi}[Security Game:]
The security game between a challenger $\C$ and an adversary $\A$ is defined as follows:
\begin{itemize}
 \item $\C$ runs the preliminary phase and sets up the parameters, $\C$ controls both $P_1$ and $P_3$.
 \item $\C$ and $\A$ participate in the key-generation algorithm.
 \item $\A$ chooses adaptively some messages $m_1,...,m_l$ for some $l>0$ and asks for a signature on them.
 $\A$ could either participate in the signature or it can query to $\C$ a signature generated by $P_1$ and $P_3$.
 \item Eventually $\A$ outputs a new message $m \ne m_i \forall i$ and a valid signature for it with probability at least $\e$.
\end{itemize}
\end{defi}

If we denote with $\tau_\A$ the adversary's tape and with $\tau_i$ the tape of the honest player $P_i$ we can write:
 \begin{equation}\label{avversario}
 	 \Pb_{\tau_i, \tau_\A}[\A(\tau_\A)_{ P_i(\tau_i)} = \mathtt{forgery}] \ge \e,
 \end{equation}
 where $\Pb_{\tau_i, \tau_\A}$ means that the probability is taken over the random tape $\tau_\A$ of the adversary and the random tape $\tau_i$ of the honest player, while $\A(\tau_\A)_{P_i(\tau_i)}$ is the output of the iteration between the adversary $\A$, running on tape $\tau_\A$, and the player $P_i$, running on tape $\tau_i$ .
 We say that an adversary's random tape $\tau_\A$ is good if:
 \begin{equation}\label{defbuono}
  \Pb_{\tau_i}[\A(\tau_\A)_{ P_i(\tau_i)} = \mathtt{forgery}] \ge \frac{\e}{2}.
\end{equation}
Now we have the following Lemma, introduced in~\cite{gennaro2018fast}:
\begin{lem} \label{goodTape}
	If $\tau_\A$ is chosen uniformly at random, then the probability of choosing a good one is at least $\frac{\e}{2}$.
\end{lem}

\begin{proof}
	In the proof we will simplify the notation writing $\A(\tau_\A,\tau_i) = \mathtt{forgery} $ instead of ${\A(\tau_\A)_{P_i(\tau_i)} = \mathtt{forgery}}$. Moreover we write $b$ to identify a good tape, while $c$ will be a bad one.
	We can rewrite Equation	\ref{avversario} in this way:
	\begin{align}
	 A &= \Pb_{\tau_i,\tau_\A}(\tau_\A = b, A(\tau_\A, \tau_i)= \mathtt{forgery}) + \Pb_{\tau_i,\tau_\A}(\tau_\A = c, A(\tau_\A, \tau_i) = \mathtt{forgery})\nonumber\\
	 &= \Pb_{\tau_\A, \tau_i} (\tau_\A = b) \Pb_{\tau_i,\tau_\A}(A(\tau_\A, \tau_i) = \mathtt{forgery}|\tau_\A = b)\nonumber\\
	 &\phantom{=}\, + \Pb_{\tau_\A, \tau_i}(\tau_\A = c) \Pb_{\tau_i,\tau_\A}(A(\tau_\A, \tau_i) = \mathtt{forgery} |\tau_\A = c).
	\end{align}
	Trivially we have that $\Pb_{\tau_i,\tau_\A}(A(\tau_\A, \tau_i) = \mathtt{forgery}|\tau_\A = b) < 1,$ and from the definition of good tape in Equation \ref{defbuono} we get:
	\begin{equation}
	 \Pb_{\tau_i,\tau_\A}(A(\tau_\A, \tau_i) = \mathtt{forgery}|\tau_\A = c) < \frac{\e}{2}.
	\end{equation}
	Now we want to solve for $x = \Pb_{\tau_\A, \tau_i} (\tau_\A = b)$, so we get:
	\begin{equation}
	 \e \le A < x \cdot 1 + (1 - x) \cdot \frac{\e}{2} = x \left(1 - \frac{\e}{2}\right) + \frac{\e}{2},
	\end{equation}
that leads us to the conclusion:
\begin{equation}
 x \ge \frac{\e - \frac{\e}{2}}{1 - \frac{\e}{2}} \ge \frac{\e}{2-\e} \ge \frac{\e}{2}.
\end{equation}
\end{proof}

From now on we will suppose that the adversary is running on a good random tape.

Now we describe the simulation for the key-generation protocol. The forger $\F$ receives from its challenger the public key $\mathcal{Y}_c=x\B$ for the centralised ECDSA, a public key $E_1$ for Paillier and a public key $\pk_3$ for the asymmetric encryption scheme.

The simulation proceeds rewinding $\A$ and repeating the following steps until the public key has been correctly generated (i.e. $\mathcal{Y} \ne \perp$), which happens if $\A$ sends a correct decommitment on behalf of $P_2$ also after the rewind of step~\ref{setup_rewind}:
\begin{enumerate}
	\item $\A$ computes and broadcasts $\mathsf{KGC}_2$ and $\mathsf{KGCS}_2$.
	\item $\F$ selects random values $u_1, \sigma_{3,1}$, computes $[\mathsf{KGC}_1,\mathsf{KGD}_1]=\Com(u_1\B)$, $[\mathsf{KGCS}_1,\mathsf{KGCS}_1]=\Com(\sigma_{3,1}\B)$and sends $\mathsf{KGC}_1, \mathsf{KGCS}_1$.
	\item Each player $P_i$ broadcasts its decommitments and the Feldman-VSS values. Let $u_i\B, \sigma_{3,i}\B$ be the values decommitted.
	\item Each player sends their Paillier public key $E_i$.
	\item\label{setup_rewind} at this point $\F$ computes $\hat{u}\B = \mathcal{Y}_c - u_2\B - 3\sigma_{3,1}\B + 2\sigma_{3,2}\B$ and rewinds $\A$ to its commitment of $\mathsf{KGC}_1$ (step 2). We remark that $\F$ does not know $\hat{u}$ but only $\hat{u}\B$.
	\item $\F$ sends $\hat{\mathsf{KGC}_1}$, the commitment corresponding to $\hat{u}\B$. The commitment for $\sigma_{3,1}\B$ remains the same.
	\item if $\A$ refuses to decommit then $u_i\B$ is set to $\perp$.
	\item $\F$ simulates the VSS (since $\F$ is not able to compute the random polynomial $f(x)$) as explained in \Cref{feldman}.
	\item $\F$ participates in the $ZK$, extracting the secret values from $\A$.
	\item the remaining steps remain the same. If $u_i=\perp$ for some $i$ the public key $y$ is set to $\perp$. It is important to note that $P_1$ does not know the value of $\sigma_{1,3}$ and therefore uses a random value to compute $\rec_{1,3}$.
\end{enumerate}

Now we prove that the simulation terminates in expected polynomial time, that it is indistinguishable from the real protocol and that it terminates with output $\mathcal{Y}_c$ except with negligible probability.

\begin{lem} \label{lemma32}
	The simulation terminates in expected polynomial time and it is indistinguishable from the real protocol.
\end{lem}
\begin{proof}
	Since $\A$ is running on a good random tape we know that it will correctly decommit with probability at least $\frac{\e}{2}$, then we need to rewind only a polynomial number of times, since the expected number of iterations we need to perform is $\frac{2}{\e} < 2 \lambda^c$.
	The first difference between the real protocol and the simulated one is that $\F$ does not know the discrete logarithm of $\hat{u}\B$ and so it needs to perform a ``fake'' Feldman-VSS. This is indistinguishable from a real Feldman-VSS since they have both the same distribution, as shown in \Cref{feldman}.
	The other difference is that $\F$ does not know the value of $\sigma_{1,3}$ and therefore uses a random value to compute $\rec_{1,3}$, but since the encryption algorithm is IND-CPA for the adversary this ciphertext is indistinguishable from the real one.
\end{proof}
\begin{lem} \label{lemma33}
	For a polynomially large fraction of possible values for the input $\mathcal{Y}_c$ the simulation terminates with output $\mathcal{Y}_c$ except with negligible probability.
\end{lem}
\begin{proof}
	First we prove that if the simulation terminates correctly (i.e. with output which is not $\perp$) then it terminates with $\mathcal{Y}_c$ except with negligible probability.
	
	This is a consequence of the non-malleability property of the commitment scheme. Indeed, if $\A$ correctly decommits $\mathsf{KGC}_2$ twice it must do so to the same string, no matter what $P_1$ decommits to (except with negligible probability). Therefore, due to our choice for $\hat{u}\B$ we have that the output is exactly $\mathcal{Y}_c$.
	
	Now we prove that the simulation ends correctly for a polynomially large fractions of the inputs.
	
	Since $\A$ is running on a good random tape it decommits correctly for at least a fraction $\frac{\e}{2}>\frac{1}{2\lambda^c}$ of the possible values of $\hat{u}\B$.
	Moreover, since $\mathcal{Y}_c$ and $ \sigma_{3,1}$ are chosen uniformly at random, and $y_2$ and $\sigma_{3,2}$ are chosen by $\A$ without the knowledge of the other values, we can conclude that ${\hat{u}\B=\mathcal{Y}_c - u_2\B - 3\sigma_{3,1} + 2\sigma_{3,2}}$ (that is fully determined before the rewind) has also uniform distribution.
	Then given the 1-to-1 correspondence between $\mathcal{Y}_c$ and $\hat{u}\B$ we can conclude that for a fraction $\frac{\e}{2}>\frac{1}{2\lambda^c}$ of the inputs the protocol will correctly terminate.
\end{proof}
\begin{oss} \label{rush}
	In the simulation it is crucial that the adversary broadcasts $\mathsf{KGC}_2$ and $\mathsf{KGCS}_2$ before $\F$. Inverting the order will cause this simulation to fail, since after the rewind $\A$ could change its commitment. Due to the non-malleability property we are assured that $\A$ can not deduce anything about the content of these commitments, but nevertheless it could use it as a seed for the random generation of its values. In this case $\F$ guesses the right $\hat{u}\B$ only with probability $\frac{1}{q}$ where $q$ is the size of the group, so the expected time is exponential.

	It is possible to swap the order in the first step using an equivocable commitment scheme with a secret trapdoor. In this case we only need to rewind at the decommitment step, we change $\mathsf{KCD}_1$ in order to match $\hat{u}\B$. In this way we could prove the security of the protocol also in the presence of a \emph{rushing adversary} but we need an additional hypothesis regarding the commitment scheme.
\end{oss}
Now we describe the simulation of the protocol for the signature generation. In the same way Gennaro and Goldfeder did in~\cite{gennaro2018fast}, we need to distinguish two different types of executions depending on what happens during the step 3 in the signature generation algorithm:
\begin{enumerate}
	\item Semi-correct executions: the adversary has followed the protocol and committed-decommitted the correct $k_2$, then the equality $\ \mathcal{R}=k^{-1}\B$ holds, where $k=k_1+k_2$. In this case $\F$ is able to correctly determine the value $s_1 \mathcal{R}$ and therefore to correctly terminate the simulation.
	\item Non-semi-correct executions: the value decommitted by $\A$ is not the correct $k_2$, then the value $k^{-1}$ is not the proper one. We will show that a simulation that is not semi-correct will fail with high probability since the value $\mathcal{U}_1$ contributed by $P_1$ is indistinguishable from a random one. This allows us to simulate the protocol by simply using a random $\tilde{s}_1$ for $P_1$ instead of the correct one.
\end{enumerate}

We note that it is impossible to distinguish the two cases a priori, so the idea is to guess if an execution will be semi-correct or not. In the semi-correct executions the simulator will be able to extract the ``signature shard'' of $P_1$ and to terminate the simulation successfully, in the non-semi-correct ones our simulator will guess a random signature causing the protocol to abort (we will show that in the non-semi-correct execution the real protocol aborts with high probability).

We now present the simulation for a semi-correct execution. We recall that $\F$ does not know the secret values of $P_1$ ($\w_1$ and the secret key corresponding to its Paillier public key) but it knows the secret values of $P_2$. In the following simulation $\F$ aborts whenever the original protocol is supposed to abort.
\pagebreak[3]
\begin{enumerate}
	\item Both players pick randomly $k_i, \gamma_i$ and broadcast $\Delta_i$.
	\item Both players execute the $\mathsf{MtA}$ protocol for $k\gamma$. Since $P_1$ can not decrypt $\alpha_{1,2}$, $\F$ sets it at random.
	\item Both players execute the MtAwc protocol for $k\w$. From the ZKP $\F$ extracts $\nu_{1,2}$. Since $\F$ does not know $\w_1$ it sends a random $\mu_{2,1}$ to $P_2$. At this point $\F$ knows $\sigma_2$.
	\item Both players execute the protocol, revealing $\delta_i$ and setting ${\delta=\delta_1+\delta_2}$.
	\item Both players send $D_i$. 
	\item $\F$ queries its signature oracle and receives a signature $(r,s)$ for $m$. It computes $\mathcal{R}=ms^{-1}\B+rs^{-1}\mathcal{Y}$. We note that, in centralised ECDSA, we have $s\mathcal{R} = k (m + r u)\mathcal{R} = m\B + ru\B = m\B + r\mathcal{Y}$. Then the $\mathcal{R}$ we compute in our simulation is the correct value in the centralised algorithm.
	\item $\F$ rewinds $\A$ and changes its commitment to $\hat{\mathcal{G}_1} = \delta ^ {-1} \mathcal{R} - \mathcal{G}_2$. In this way we have that $\delta(\hat{\mathcal{G}_1} + \mathcal{G}_2) = \mathcal{R}$.
	In the steps after the new commitment $P_1$ uses the old values of $\gamma_1$ and $k_1$ since it does not know the correct ones.
	\item At this point $\F$ knows the value $s_2$ and then it can compute the right $s_1$ as $s_1 = s-s_2$. 
	\item $P_1$ and $P_2$ follow the remaining part of the protocol normally.
\end{enumerate}
\begin{oss}
	As in the case of the Enrollment Phase we need that $\A$ speaks first, in order to rewind to our commitment phase without changing $\A$'s random tape. The same argument about equivocable commitment schemes and rushing adversaries applies here as well.
\end{oss}
\begin{lem}
	Assuming that
	\begin{itemize}
		\item The Strong RSA Assumption holds,
		\item we use a non malleable commitment scheme,
	\end{itemize}
then the simulation has the following properties:
\begin{itemize}
	\item on input $m$ it outputs a valid signature $(r,s)$ or aborts,
	\item it is computationally indistinguishable from a semi-correct real execution.
\end{itemize}
\end{lem}
\begin{proof}
	The differences between the real and the simulated views are the following:
	\begin{itemize}
		\item $P_1$ does not know the discrete logarithm of $\mathcal{G}_1$. Moreover $\mathcal{G}_1 \ne \gamma_1 \B$.
		\item in the real protocol $\mathcal{R} = (k_1+k_2)^{-1} \B$, in this simulation $\mathcal{R}$ is chosen by the signing oracle.	
	\end{itemize}
We have that $c_1 = E(\gamma_1)$ is sent during the MtA protocol. In order to distinguish between a real execution and a simulated one an adversary should detect if $c_1$ is the encryption of a random plaintext or if it is the encryption of $\log_\B(\mathcal{G}_1)$. The strong RSA assumption assures the semantic security of Paillier's encryption, so this is infeasible.

In the same way let $e_1 = E(k_1)$. We can write $\mathcal{R}=(\hat{k}_1+k_2)^{-1}\B$, where we could imagine $\hat{k}_1$ sent by a random oracle (in reality we never calculate $\hat{k}_1$ in the simulation because we compute directly $\mathcal{R}$). Then $(\hat{k}_1+k_2) \mathcal{R} = \B$, so $\hat{k}_1\mathcal{R} = \B -k_2\mathcal{R}$ and $\hat{k}_1=\log_\mathcal{R}(\B-k_2\mathcal{R})$. With the same argument as before we can conclude that the simulation is indistinguishable from the real execution.
It is worth noting that we are simulating a semi-correct execution with a non-semi-correct one, but since they are indistinguishable this is fine.

Now let $(r,s)$ be the signature that $\F$ receives by its oracle in the sixth point of the protocol.
Note that the change of the commitment after the rewind does not change the view for $\A$ given the hiding and non-malleabilty properties of the commitment scheme and the considerations above, so as a consequence of the non-malleability property of the commitment scheme the decommitment is consistent and we have that if the protocol terminates it does so with output $(r,s)$.
\end{proof}

Now we show how to simulate the protocol for a non semi-correct execution, i.e. when the value decommitted by $\A$ is not the real $k_2$.
\begin{enumerate}
	\item the simulator runs the semi-correct simulation from the first to the sixth point.
	\item $\F$ does not rewind $\A$ to fix the value of $\mathcal{R}$. Instead it runs the protocol normally.
	\item $\F$ chooses $\tilde{s_1} \in \mathbb{Z}_q$ and $\mathcal{U}_1$ at random and uses these values in the last part of the protocol.
\end{enumerate}
The only difference between the semi-correct simulation and this one is the choice of $s_1$ and $\mathcal{U}_1$. The reason is that in the semi-correct simulation $\F$ can arbitrarily fix $\mathcal{R}$ because the values decommitted by $\A$ are the real ones, instead in the non-semi-correct simulation this is impossible since the value $k_1$ does not match anymore. Therefore $\F$ tries to make the protocol fail choosing $\mathcal{U}_1$ and $s_1$ at random.

We divide the proof in two different steps: first we prove that a real non-semi-correct execution is indistinguishable from a simulation in which $P_1$ uses the right $s_1$ but outputs a random $\mathcal{U}_1$ and then we prove that this second intermediate simulation is indistinguishable from the one described above, with both $s_1$ and $\mathcal{U}_1$ chosen randomly.
\begin{lem}
	Assuming that 
	\begin{itemize}
		\item the DDH Assumption holds,
		\item $\Com, \Ver$ is a non malleable commitment scheme,
	\end{itemize}
then the simulation is computationally indistinguishable from a non-semi-correct real execution.
\end{lem}
\begin{proof}
	As anticipated we construct three games between $\F$ and $\A$. In the first one, $G_0$, the simulator will simply run the real protocol. In the game $G_1$, $\F$ follows the real protocol but chooses $\mathcal{U}_1$ randomly. Finally, in $G_2$, $\F$ runs the simulation previously described, with both $s_1$ and $\mathcal{U}_1$ chosen at random. Now we proceed to prove the indistinguishability of $G_0$ and $G_1$ and then of $G_1$ and $G_2$.
	\\
	
	Let us assume that there is an adversary $\mathscr{A}_0$ that can distinguish between $G_0$ and $G_1$. We show that this contradicts the DDH Assumption.
	\\Let $\tilde{\mathcal{A}} = a\B, \tilde{\mathcal{B}} = b\B, \tilde{\mathcal{C}} = c\B$ be the DDH challenge where $c=ab$ or $c$ is random in $\Z_q$.
	The distinguisher $\F_0$ runs $\mathscr{A}_0$, simulating the key-generation phase so that $\mathcal{Y}= b\B$. It can do that by rewinding the adversary and changing its decommitment to $u_1\B = \mathcal{Y}-u_2\B-\sigma_{3,1}\B + 2\sigma_{3,2}\B$, making $\mathcal{Y}=\tilde{\mathcal{B}}$. Thanks to the ZKP $\F_0$ extracts the values of $x_2$ (and then of $\w_2$) from the adversary, but does not know $b$ (and therefore not $x_1$ nor $\w_1=b-\w_2$). In this simulation we can also suppose that $\F_0$ knows the secret key associated to $E_1$, its public key of the Paillier cryptosystem (we can do this since we are not making any reduction to the security of the encryption scheme).
	
	At this point $\F_0$ runs the signature generation protocol for a non-semi-correct execution. It runs the protocol normally till the MtA and MtAwc part of the signature protocol. It knows $\gamma_1$, $k_1$, since it runs $P_1$ normally, and $\gamma_2$, since it extracts its value from the adversary. Therefore $\F_0$ knows $k$ such that $\mathcal{R}=k^{-1}\B$. Since we suppose that $\F_0$ knows the secret key associated to $E_1$ it can also know $\mu_{1,2}$ obtained from the $\MtAwc$ protocol on input $\w_2$ and $k_1$. Since it does not know $\w_1$, during the $\MtAwc$ protocol on input $\w_1$ and $k_2$ it sends a random $\mu_{2,1}$ and sets:
	\begin{equation}\label{ECnsc.eq.1}
	 \nu_{2,1}=k_2\w_1-\mu_{2,1}.
	\end{equation}
	So at the end of the $\MtAwc$ we have that:
	\begin{equation}\label{ECnsc.eq.2}
	 \sigma_1 = k_1\w_1 + \mu_{1,2} + \nu_{2,1};
	\end{equation}
	using \cref{ECnsc.eq.1} in \cref{ECnsc.eq.2} we get
	$$\sigma_1 = k_1\w_1 + k_2\w_1 - \mu_{2,1} + \mu_{1,2} = \tilde{k}\w_1 + \mu_{1,2} - \mu_{2,1},$$
	where $\tilde{k} = k_1 + k_2$ and we have that $\tilde{k} \ne k$ since we are in a non-semi-correct execution.
	Remembering that $\w_1 = b - \w_2$ we can substitute again, obtaining
	\begin{equation}
	 \sigma_1 = \tilde{k}b - \tilde{k}\w_2 + \mu_{1,2} - \mu_{2,1},
	\end{equation}
	and $\F_0$ knows every value in the equation except $b$, so let us group these known values and set $\mu_1 = \tilde{k}\w_2 + \mu_{1,2} - \mu_{2,1}$. Thus $\F_0$ can successfully compute:
	\begin{equation}
	 \sigma_1 \B = \tilde{k}b\B+ \mu_1 \B = \tilde{k}\tilde{\mathcal{B}} + \mu_1 \B,
	\end{equation}
	and therefore also:
	\begin{equation}
	 s_1\mathcal{R} = (k_1m+r\sigma_1) k^{-1} \B = (k_1m+r\mu_1)k^{-1}\B+\tilde{k} r k^{-1}\tilde{\mathcal{B}}.
	\end{equation}
	
We now proceed to the last part of the simulation. 
$\F_0$ selects a random $l_1$ and sets $\mathcal{W}_1 = s_1\mathcal{R}_1+l_1\B$. Instead of following the algorithm $\F_0$ does not choose a random $\rho_1$ but sets implicitly $\rho_1 = a$ and sends $\mathcal{A}_1 = a\B = \tilde{\mathcal{A}}$. During the ZKP that it simulates (since it does not know $a$ nor $s_1$) it extracts $s_2, l_2, \rho_2$ from the adversary. Let us define $s = k^{-1}s_2$. We note that:
\begin{align}
\mathcal{W} &= -m\B -r\mathcal{Y} + \mathcal{W}_1 + \mathcal{W}_2\\
&= (l_1 + l_2)\B + s_1\mathcal{R} + (s-m) \B -r \mathcal{Y}\\
&= l \B+ t_1 \B + t_2 \tilde{\mathcal{B}}, 
\end{align}
where $t_1 = k^{-1}(k_1m+r\mu_1) + s - m$ and $t_2 = k^{-1}\tilde{k}r - r$. We note that in a not-semi-correct execution $\tilde{k} \ne k$ and then $t_2 \ne 0$.
Finally $\F_0$ computes $\mathcal{T}_1 = l_1\mathcal{A}$ correctly but for $\mathcal{U}_1$ it outputs $\mathcal{U}_1 = (l+t_1)\tilde{\mathcal{A}} + t_2\tilde{\mathcal{C}}$ and aborts.
\\More explicitly we have that:
\begin{align}
	\mathcal{U}_1 &= (l+t_1)\tilde{\mathcal{A}}+ t_2\tilde{\mathcal{C}}\\
	&= (l+k^{-1}(k_1m+r\mu_1)+s-m)\tilde{\mathcal{A}} + (k^{-1}\tilde{k}r-r)\tilde{\mathcal{C}}\\
	&=\left(l_1+l_2 + k^{-1}(k_1m+r\mu_1)+s-m\right)a\B + (k^{-1}\tilde{k}r - r)c\B
\end{align}
If we have $c=ab$ this equation can be further simplified to:
\begin{align}
	\mathcal{U}_1 
	&= a(l_1+l_2)\B + ak^{-1}(k_1m+r\mu_1+\tilde{k}rb) + a(s-m)\B - abr\B\\
	&= a(l_1+l_2)\B + ak^{-1}s_1\B+a k^{-1}s_2\B - am\B - abr\B\\
 &= a(l_1 + l_2) \B + a k^{-1}(s_1 + s_2) \B - a(br+m)\B,
\end{align}
and we note also that:
\begin{align}
 a\mathcal{W} &=a \mathcal{W}_1 + a \mathcal{W}_2 -am\B -ar\mathcal{Y}\\
 &= as_1 \mathcal{R} + al_1\B + as_2 \mathcal{R} + al_2\B -am\B -arb\B\\
 &= a(l_1 + l_2) \B + a k^{-1}(s_1 + s_2) \B - a(br+m)\B.
\end{align}
Then we can conclude that if $c=ab$ we have $\mathcal{U}_1 = aV = \rho_1 \mathcal{W}$, as in $G_0$, otherwise $\mathcal{U}_1$ is a random group element as in $G_1$. Then if a distinguisher for $G_0$ and $G_1$ exists it can be also used to win a DDH challenge as we described above, so $G_0$ and $G_1$ are indistinguishable.
\\

Now we deal with the indistinguishability of $G_1$ and $G_2$. We recall that the difference between the protocols $G_1$ and $G_2$ is that in $G_2$ we use a random $\tilde{s}_1$ during the last part and then we have a random $\tilde{\mathcal{W}}_1 = {\tilde{s}_1}\mathcal{R} + l_1\B$. Once again we will prove that $G_1$ is indistiguishable from $G_2$ performing a reduction to the DDH assumption. The idea is to show that $\mathcal{W}_1 = s_1\mathcal{R}+l_1\B$ and $\tilde{\mathcal{W}}_1=\tilde{s}_1\mathcal{R}+l_1\B$ are indistinguishable due to the random value $l_1\B$ added to both of them.

Let $\tilde{\mathcal{A}} = (a-d)\B$, $\tilde{\mathcal{B}} = b\B$, $\tilde{\mathcal{C}} = ab\B$ be the DDH challenge where $d = 0$ or $d$ is a random value of $\Z_q$. The simulator proceeds with the regular protocol until step~\ref{ECsgigen}\ref{ECsigcheck_com} Then $\F_0$ broadcasts $\mathcal{W}_1 = s_1\mathcal{R} + \tilde{\mathcal{A}}$ and $\mathcal{A}_1 = \tilde{\mathcal{B}}$. It simulates the ZKP (since it does not know $l_1$ and $\rho_1$) and extracts $s_2,$ $l_2$, $\rho_2$. Then it computes $\mathcal{U}_1$ as a random element and $\mathcal{T}_1 = \tilde{\mathcal{C}}+\rho_2\tilde{\mathcal{A}}=ab\B + \rho_2(a-d)\B$.

When $d=0$ we have $\tilde{\mathcal{A}} = a\B$, so $a = l_1$ , $b = \rho_1$ and we have that $\mathcal{W}_1 = s_1R+l_1\B$ and:
\begin{align}
 \mathcal{T}_1
 = ab\B +\rho_2 a \B
 =l_1 \rho_1 \B + l_1 \rho_2 \B
 =l_1(\rho_1 + \rho_2) \B,
\end{align}
so $\mathcal{T}_1 = l_1\mathcal{A}$ as in $G_1$. Otherwise, when $d \ne 0$ we have that $\tilde{\mathcal{A}}=a\B-d\B$ with a randomly distributed $d$, then this is equivalent to have:
\begin{equation}
 \mathcal{W}_1 =s_1\mathcal{R}+(a-d)\B = \tilde{s}_1\mathcal{R} +a\B,
\end{equation}
with $\tilde{s}_1 = s_1-d k^{-1}$, that is uniformly distributed thanks to $d$, and $\mathcal{T}_1 = l_1 \mathcal{A}$ as in $G_2$. The key idea is that we use the random value $d$ to change the fixed value $s_1$ to a random and unknown $\tilde{s}_1$ during the computation of $\mathcal{W}_1$. Therefore, under the DDH assumption, $G_1$ and $G_2$ are indistinguishable. Then $G_0$ is indistinguishable from $G_2$ as we wanted to prove.
\end{proof}

Now we have to deal with the recovery signature algorithm. Since the core algorithm remains the same we can use the two proofs already explained, we only need to change the setup phase in which the third player recovers its secret material. In this section we still examine the case in which $\A$ controls one between $P_1$ or $P_2$ and $\F$ controls $P_3$. Things are a little bit different if $\A$ controls $P_3$ since it does not perform the enrollment phase (this case is much easier).
\\Trivially if $\A$ asks for a recovery signature between the two honest parties $\F$ can simply query its oracle and output whatever the oracle outputs. So we can limit ourselves to deal with the case where the adversary participates in the signing process.
\\
\\
Without loss of generality we suppose that $\A$ controls $P_2$. The simulation works as follows:
\begin{enumerate}
	\item $P_2$ sends $\mathcal{Y}, \rec_{1,2}, \rec_{1,3}$ to $P_3$.
	\item $P_i$ generates a Paillier public key $(N_i, \Gamma_i)$ and sends it to the other party.
	\item $P_i$ proves in ZK that it knows the matching secret for its public Paillier key.
	\item $P_3$ can not decrypt the values received in step 1, so it simulates the ZKP about $x_3$ and at the same time it can extract the secret value $x_2$ from $P_2$. Note that the inability to decrypt is not a problem since most of the data is useless (the random values sent by $\F$ during the enrolment phase), so it would not have been able to compute $x_3$ nor $\w_3$ anyway.
	However the simulator remembers the correct values of $\rec_{1,2}, \rec_{1,3}$, so if $\A$ does not send the proper ones it can abort the simulation, as it would happen in a real execution since $P_3$ can not recover a secret key shard that matches $\mathcal{Y}$.
	\item $P_2$ computes its $\tilde{x}_2$ and $\tilde{\w}_2$ from its original shards.
	\item $P_2$ and $P_3$ perform the signing algorithm with the above simulation. Also in this case $\F$ does not know its secret key, but we remark that this is fine since it can use the signing oracle.
\end{enumerate}
In the case of $P_3$ being dishonest the simulation is much more easier. During the enrolment phase $\F$ can produce random shards to send to $P_3$ during the recovery signature phase and output directly its original ECDSA challenge. Then with the same algorithm as before it can perform the signing protocol.
\\
\\
Now we are ready to prove \Cref{theorem_security}.
\begin{proof}
	Let $Q<\lambda^c$ be the maximum number of signature queries that the adversary makes. In the real protocol the adversary will output a forgery after $l<Q$ queries, either because it stops submitting queries or because the protocol aborts. In our simulation we try to guess if a simulation will be semi-correct or not choosing a random $i \in [0,Q]$. We have two cases:
	\begin{itemize}
		\item if $i=0$ we assume that all the executions are semi-correct and then we always use the semi-correct algorithm described previously.
		\item if $i \ne 0$ we assume that the first $i-1$ are semi-correct, but the $i^{th}$ is not. In this case we use the semi-correct algorithm for every execution except for the $i^{th}$ one, for which we use the non-semi-correct one, then we abort.
	\end{itemize}
As we previously proved we produce an indistinguishable view for the adversary, then $\A$ will produce a forgery with the same probability as in a real execution. Then the probability of success of our forger $\F$ is $\frac{\e^3}{8Q}$ and it is the product of:
\begin{itemize}
	\item the probability of choosing a good random tape for $\A$, that is at least $\frac{\e}{2}$, a shown in Lemma \ref{goodTape},
	\item the probability of hitting a good public key, that also is at least $\frac{\e}{2}$, as shown in Lemma \ref{lemma32} and Lemma \ref{lemma33},
	\item the probability of guessing the right index $i$, that is $\frac{1}{Q}$,
	\item the probability of $\A$ to successfully produce a forgery on a good random tape, that is $\frac{\e}{2}$ as shown in Equation \ref{defbuono}.
\end{itemize}
Under the security of the ECDSA signature scheme the probability of producing a forgery must be negligible, which implies that $\e$ must negligible too, contradicting the hypothesis that $\A$ has non-negligible probability of forging the scheme.
\end{proof}

\section{Conclusions}\label{conclusions}
Although decentralized signature algorithms have been known for a while, we are aware of only few proposals for algorithms that are able to produce signatures indistinguishable from a standard one.
Moreover, the protocol described in this work is, as far as we know, the first example of a threshold multi-signature allowing the presence of an off-line participant.
Regarding the protocol's specification, the main difference w.r.t.~\cite{gennaro2018fast} lies in the key-generation phase.
Specifically, the idea is to have two active participants to simulate the action of the third one.
This step is possible due to the uniqueness property of polynomial interpolation that gives a bijection between points and coefficients, that combined to the preserved uniform distribution  in $\Z_p$ allows us to ``invert'' the generation of the shares, that are later recovered by the offline party thanks to an asymmetric encryption scheme.
A second divergence is that we have managed to avoid the use of equivocable commitments under the assumption that in some specific steps (see~\Cref{rush}) we can consider the adversary to not be rushing.

The main efficiency bottleneck is in the massive usage of ZKPs, which are necessary to guarantee the security of the signature itself against black-box adversaries, as hinted by the security proof of \Cref{ecdsasecurity}.
Nevertheless, there are implementations that use our protocol to resiliently manage bitcoin wallets~\cite{dinicola_2020}.

In our security analysis we focused on the unforgeability of the signature, however with an offline party (and more so in the application context of crypto-assets management) there is another security aspect worthy of consideration: the resiliency of recovery in the presence of a malicious adversary.
Of course if the offline party is malicious and unwilling to cooperate in fund recovery there is nothing we can do about it, however the security can be strengthened if we consider that one of the online parties may corrupt the recovery material.
In this case a generic CPA asymmetric encryption scheme is not sufficient to prevent malicious behaviour, because we need a verifiable scheme that allows the parties to prove that the recovery material is consistent, just like they prove that they computed the shards correctly. In this way the adversary is unable to corrupt the recovery material and then there is always a pair of players that is able to sign.
Indeed without these protection measures a malicious user could convince an unsuspecting victim to set up a $(2,3)$-threshold wallet together, and sending bogus recovery data the attacker can later on blackmail the victim to sign transactions of their choice otherwise it will refuse to collaborate in future signatures, effectively freezing the funds since the recovery party has been neutralised.

An interesting topic of further analysis surely regards provable public-key encryption schemes, for example we see potential solutions that could exploit the homomorphic properties of Paillier or ElGamal~\cite{elgamal1985public} cryptosystems.

Other future research steps involve the generalisation to $(t,n)$-threshold schemes with more than one offline party, as well as to different standard signatures.
Regarding the latter, there is a variant of our protocol whose signatures are indistinguishable from EdDSA~\cite{longo2020threshold}, and we are working on its security proof.

\paragraph{Acknowledgments}
The core of this work is contained in the first author's MSC thesis that would like to thank his two supervisors, the second and fourth author, and Telsy S.p.A. for their support during the work.
Part of the results presented here have been carried on within the EU-ESF activities, call PON
Ricerca e Innovazione 2014-2020, project Distributed Ledgers for Secure Open Communities. The second and third authors are members of the INdAM Research group GNSAGA. 
\\
We would like to thank Conio s.r.l. and its co-CEO Vincenzo di Nicola for their support. We also thank Gaetano Russo, Federico Mazzone, and Zsolt Levente Kucsv\'an that worked on the implementation and provided valuable feedback.
\\
The authors would like to thank the anonymous referees.

\nocite{*} 
\printbibliography

\appendix

\section{Zero Knowledge Proofs}

\subsection{Schnorr Protocol}\label{Schnorr}
The Schnorr Protocol is a zero-knowledge proof for the discrete logarithm.
\\Let $\mathbb{G}$  be a group of prime order $p$ with generator $g$. Let $h\in \mathbb{G}$ be a random element in $\mathbb{G}$.
The prover $\mathscr{P}$ wants to prove to a verifier $\mathscr{V}$ that it knows the discrete logarithm of $h$, i.e. it knows $x\in \Z_p$ such that $g^x=h$.\\
So the common inputs are $\mathbb{G},g$ and $h$, while the secret input of $\mathscr{P}$ is $x$.\\
The protocol works as follows:
\begin{enumerate}
	\item $\mathscr{P}$ picks $r$ uniformly at random in $\Z_p$ and computes $u=g^r$. It sends $u$ to $\mathscr{V}$.
	\item $\mathscr{V}$ picks $c$ uniformly at random $\in \Z_p$ and sends it to $\mathscr{P}$.
	\item $\mathscr{P}$ computes $z=r+cx$ and sends $z$ to $\mathscr{V}$.
	\item $\mathscr{V}$ computes $g^z$. If $\mathscr{P}$ really knows $x$ it holds that $g^z=uh^c$. If the equality does not hold, the verifier rejects.
\end{enumerate}
A detailed proof about the security of the algorithm can be found in \cite{schnorr1989efficient}.

\subsubsection{Schnorr Protocol Simulation}
We need to simulate the Schnorr protocol in two different ways: first we need to use it to extract the adversary's secret value, then we need to simulate it without knowing our secret value, tricking the opponent.
We can use the Schnorr protocol to extract the value $x$ from the adversary in this way:
    \begin{enumerate}
        \item Follow the standard protocol until the third point, obtaining $z$.
        \item Rewind the adversary to the second point and pick $c'\ne c$.
        \item Follow the remaining part of the protocol, obtaining $z'$.
        \item We can compute $\frac{z-z'}{c-c'}=\frac{(c-c')x}{c-c'}=x$.
    \end{enumerate}
    \begin{proof}[Sketch]
    Since the only extra hypothesis for $c'$ is that $c'\ne c$ we can suppose that $c'$ has uniform distribution as well. Moreover $z$, once the verifier sent $c$ the value of $z$ is fixed, so the rewinding technique does not cause any problem.
    \end{proof}
    
      At the same time we need to be able to simulate the protocol without knowing $x$. The simulation works as follows:
    \begin{enumerate}
        \item Follow the protocol until the second point, obtaining $c$.
        \item Rewind the adversary to the first point. The simulator picks $r$ randomly and computes $u'=g^{-xc+r} = (g^x)^{-c}g^r$. Under the discrete logarithm assumption and since $r,c$ are random element, this is indistinguishable from $g^r$.
        \item The simulator sends $u'$ and the adversary sends $c$ again.
        \item The simulator sends $z=r-cx+cx = r$.
        \item The adversary checks that $g^z=g^r=u'(g^x)^{c}=g^{-x c}g^r g^{x c}.$
    \end{enumerate}
    \begin{proof}[Sketch]
    The tricky point of the simulation is the third point, when we need that the adversary sends the same $c$ it has previously sent, since sending a different $r$ could change the random choice of $c$ . This could be achieved introducing an equivocable commitment scheme, in this way we need only to change the decommitment value after receiving the adversary commitment.
    \end{proof}

\subsection{Integer Factorization Proof}\label{fact}
We now present a well-known ZK proof for the integer factorization problem.
\\Let $k$ be the security parameter and $N = pq$ with neither $p$ nor $q$ small.
Let $A, B, l$ be such that \footnote{$\phi$ is the Euler's totient function that counts the positive integers up to a given integer n that are relatively prime to n, while $\theta$ is the Bachmann-Landau symbol that means that the argument is bounded both above and below.} $l \log B = \theta(k)$, $(N - \phi(N))l B < A < N$.
Let $z_1, \ldots, z_{k} \in \Z^*_N$ be chosen uniformly at random.
The protocol consists in the repetition $l$ times of the following sub-protocol:
\begin{enumerate}
    \item The prover $\mathscr{P}$ picks $r \in \{0,...,A-1\}$ at random and computes $x_i = z_i^r \mod N$ for all $i\in\{1,\ldots,k\}$.
    \item $\mathscr{P}$ sends every $x_i$ to the verifier $\mathscr{V}$.
    \item $\mathscr{V}$ picks a random integer $e \in \{0,...,B-1\}$ and sends it to $\mathscr{P}$.
    \item $\mathscr{P}$ computes $y = r+e(N-\phi(N))$ and sends it to $\mathscr{V}$.
    \item $\mathscr{V}$ checks that $0\le y < A$ and that $z_i^{y-N e} = x_i \mod N$ for all $i$. If it does not hold the protocol fails.
\end{enumerate}
A detailed explanation of the protocol and the security proof, as well as proofs about the soundness and the completeness of it, can be found in~\cite{integer}.

\subsubsection{Integer Factorization Proof}
    We need to simulate the protocol without knowing the factorization of $n$.
    A detailed description and proof of a simulation can be found in~\cite{integer}.
    Alternatively, at the price of an equivocable commitment, we could simplify the simulation as follows:
\begin{enumerate}
    \item Follow the point of the protocol choosing a random $e_r$ and $y_r$ and compute $x_i=z_i^{y-ne}$.
    \item Receive $e$ from the opponent. If $e=e_r$ then the simulation could end, otherwise rewind the opponent and change the pair $e_r, y_r$. Repeat this step until $e=e_r$
    \item The opponent sends $e$ again, the simulator sends $y'$.
\end{enumerate}
\begin{proof}[Sketch]
Clearly $z^{y'-ne}=x_i \mod n$ and $y'<A$ holds by construction.
\\We observe that a good pair $e_r,y_r$ is obtained with probability $\frac{1}{B}$ and so the complexity of all the simulation is $lB$.
\end{proof}

\subsection{Range Proof}\label{range}
We need two protocols to ensure that the share conversion protocol explained in \Cref{mta} will run correctly: one started by the initiator and the other started by the receiver.

For the first one the common inputs are a Paillier public key $(\Gamma, N)$, the ciphertext $c \in \Z_{N^2}$, an RSA modulus $M$ product of two safe primes, and $h_1, h_2 \in \Z^*_{M}$.
The prover knows $m \in \Z_q$ and $r \in\ \Z^*_N$ such that ${c = \Gamma^m r^N \mod N^2}$ where $q$ is the order of the group used during the share conversion protocol previously described.
At the end of the protocol the verifier is convinced that $m \in [-q^3, q^3]$.
\\The protocol works as follows:
\begin{enumerate}
    \item $\mathscr{P}$ picks randomly $\alpha \in \Z_{q^3}$, $\beta \in \Z^*_N$, $\gamma \in \Z_{q^3 M}$, $\rho \in \Z_{qM}$.
    \item $\mathscr{P}$ computes $z=h_1^mh_2^\rho \mod M$, $u= \Gamma^\alpha\beta^N \mod N^2$, and \linebreak[4]${\omega =h_1^\alpha h_2^\gamma \mod M}$.
    \item $\mathscr{P}$ sends $z,u$ and $\omega$ to $\mathscr{V}$.
    \item $\mathscr{V}$ picks $e$ at random and sends it to $\mathscr{P}$.
    \item $\mathscr{P}$ computes $s=r^e\beta \mod N$, $s_1= em+\alpha$ and $s_2=e\rho + \gamma$.
    \item $\mathscr{P}$ sends $s,s_1$ and $s_2$ to $\mathscr{V}$.
    \item $\mathscr{V}$ checks if $s_1 \le q^3, u=\Gamma^s_1s^Nc^{-e} \mod N^2$ and $h_1^{s_1} h_2^{s_2} = z^e \omega \mod M.$
\end{enumerate}

In the second protocol the prover wants to show that $|b| \le q^3$ and that it knows $b, \beta'$ such that $g^b=B$ and $\mathtt{c}_B= (b \times_E \mathtt{c}_A) +_E E_A(y)$ (this second part only during the Share Conversion Protocol with check). The common inputs are $B$, the Pallier public key $(\Gamma,N)$, and the ciphertexts $c_1,c_2$ (that are the Paillier ciphertexts $\mathtt{c}_A, \mathtt{c}_B$). $\mathscr{P}$ also knows $b \in \Z_q, y \in \Z_N$ and $c_2=c_1^b \Gamma^y r^N \mod N^2$. The protocol works as follows:
\begin{enumerate}
    \item $\mathscr{P}$ picks $\alpha \in \Z_{q^3}$, $\rho, \sigma, \tau \in \Z_{qM}$, $\rho' \in \Z_{q^3M}$, $\beta, \gamma \in \Z^*_N$ uniformly at random.
    \item $\mathscr{P}$ computes:
    \begin{itemize}
    \item $z=h_1^bh_2^\rho \mod M,$
    \item $z'=h_1^\alpha h_2^{\rho'} \mod M,$
    \item $t=h_1^y h_2^\sigma \mod M,$
    \item $u=g^\alpha,$
    \item $v=c_1^\alpha \Gamma^\gamma \beta^N \mod N^2,$
    \item $\omega = h_1^\gamma h_2^\tau \mod M$.
     \end{itemize}
    \item $\mathscr{P}$ sends $z,z',t, v, \omega, u$ to $\mathscr{V}$.
    \item $\mathscr{V}$ picks $e \in \Z_q$ uniformly at random and sends it to $\mathscr{P}$.
    \item $\mathscr{P}$ computes:
    \begin{itemize}
    \item $s=r^e\beta \mod N$,
    \item $s_1=eb+\alpha,$
    \item $s_2=e\rho + \rho',$
    \item $t_1=ey+\gamma$, 
    \item $t_2=e\sigma+\tau.$
    \end{itemize}
    \item $\mathscr{P}$ sends $s,s_1, s_2, t_1$ and $t_2$ to $\mathscr{V}$.
    \pagebreak[3]
    \item $\mathscr{V}$ checks if:
    \begin{itemize}
    \item $s_1\le q^3,$
    \item $g^{s_1}=B^eu,$
    \item $h_1^{s_1}h_2^{s_2} = z^ez' \mod M,$
    \item $h_1^{t_1}h_2^{t_2} = \omega t^e \mod M,$
    \item $c_1^{s_1}s^N\Gamma^{t_1}=c_2^ev \mod N^2.$
    \end{itemize}
\end{enumerate}
For a security proof of this protocol see \cite{gennaro2018fast}.

\end{document}